\documentclass[a4paper,USenglish,cleveref,autoref,thm-restate]{lipics-v2021}
\usepackage{mathtools}
\usepackage{cite}

\newcommand{\ang}[1]{\langle #1\rangle}

\renewcommand{\ang}[1]{\langle #1\rangle}

\newcommand{\RE}{\mathbb{R}}

\newcommand{\bd}{\partial\kern+1pt}     

\newcommand{\SP}{\kern+1pt}             
\newcommand{\PERP}{\kern-2pt \perp \kern-2pt} 
           
\DeclareMathOperator{\interior}{int}

\title{On The Heine-Borel Property and Minimum Enclosing Balls}
\titlerunning{On The Heine-Borel Property and Minimum Enclosing Balls}

\author{Hridhaan Banerjee}{Thomas Jefferson High School for Science and Technology, Virginia, USA \and \url{~}}{hridhaan.s.banerjee@gmail.com}{}{}

\author{Carmen Isabel Day}{California State University Channel Islands, California, USA  \and \url{~}}{carmen.isabel.day@gmail.com}{}{}

\author{Megan Hunleth}{Montgomery Blair High School, Silver Spring, Maryland, USA  \and \url{~}}{megan@hunleth.com}{}{}

\author{Sarah Hwang}{Department of Computer Science, University of Maryland, College Park, USA \and \url{~}}{shwang18@terpmail.umd.edu}{}{}

\author{Auguste H. Gezalyan}{Department of Computer Science, University of Maryland, College Park, USA \and \url{~}}{octavo@umd.edu}{https://orcid.org/0000-0002-5704-312X}{}

\author{Olya Golovatskaia}{Mount Holyoke College, Massachusetts, USA\and \url{~}}{olga.golovatskaia@gmail.com}{}{}

\author{Nithin Parepally}{Department of Computer Science, University of Maryland, College Park, USA \and \url{~}}{nparepa@terpmail.umd.edu}{}{}

\author{Lucy Wang}{Princeton University, New Jersey, USA \and \url{~}}{lucywangj@gmail.com}{}{}

\author{David M. Mount}{Department of Computer Science, University of Maryland, College Park, Maryland, USA \and \url{https://www.cs.umd.edu/~mount/}}{mount@umd.edu}{https://orcid.org/0000-0002-3290-8932}{}

\authorrunning{Banerjee, Day, Hunleth, Hwang, Gezalyan, Golovatskaia, Parepally, Wang, and Mount}

\Copyright{H. Banerjee, C. I. Day, M. Hunleth, S. Hwang, A. H. Gezalyan, M. Golovatskaia, N. Parepally, L. Wang, and D. Mount}



\date{\today}

\acknowledgements{The REU grant CNS-2150382 for REU-CAAR funded part of this work.}

\hideLIPIcs
\nolinenumbers
\begin{document}

\maketitle

\begin{abstract}
In this paper, we contribute a proof that the problem of determining the minimum radius balls over metric spaces with the Heine-Borel property is always LP-type. Additionally, we prove that weak metric spaces, those without symmetry, also have this property if we fix the direction in which we take their distances from the centers of the balls. We use this to prove that the minimum radius ball problem is also LP-type in the Thompson metrics and Funk weak metric. We finally examine the LP algorithm and explicit primitives for computing the minimum radius ball in the Hilbert metric. 
\end{abstract}

\section{Introduction}
The concept of an LP-type problem originates with the work of Micah Sharir and Emo Welzl in 1992 in their paper "A combinatorial bound for linear programming and related problems"\cite{sharir1992combinatorial}. Since then, a variety of problems have been shown to be LP-type, including minimum radius balls in the Euclidean metric \cite{matouvsek1992subexponential} and spaces with Bregman divergences \cite{nielsen2008smallest}, finding the closest distance between two convex polygons \cite{matouvsek1992subexponential}, various game-theoretic games \cite{halman2007simple} such as simple stochastic games and parity games. The analysis of LP-type problems has been a continuous area of study in computational geometry since their inception. We contribute to this by proving that determining the minimum radius ball in any metric space satisfying the Heine-Borel property is LP-type. We provide an example with the Hilbert and Thompson metrics. We also show that this holds for weak metric spaces, that is, where the distance function fails to be symmetric.

The Hilbert metric originates from the work of David Hilbert in 1895 in relation to Hilbert's fourth problem~\cite{hilbert1895linie}. It presents a non-Euclidean metric in which the triangle inequality is not strict. The Hilbert metric generalizes the Cayley-Klein model of hyperbolic geometry to arbitrary convex bodies in an $n$-dimensional space. Definitions will be given in Section \ref{sec:prelims}. Given a convex body, $\Omega \subset \RE^n$, the Hilbert metric has many desirable properties, such as the fact that straight lines are geodesics and are preserved under projective transformations. In the probability simplex, it provides a natural distance between discrete probability distributions, as shown by the works of Nielsen and Sun \cite{nielsen2019clustering, nielsen2022nonlinear}. For an excellent resource on Hilbert geometry, see "From Funk to Hilbert Geometry" or the "Handbook of Hilbert Geometry", both by Papadopoulos and Troyanov\cite{papadopoulos2014handbook, papadopoulos2014funk}. 

The Hilbert metric has seen recent use in a diverse set of fields, especially in that of convexity approximation. This is in particular due to its relationship with Macbeath regions (which are equivalent to Hilbert balls up to a scaling factor) \cite{abdelkader2018delone,abdelkader2024convex} and the flag approximability of polytopes \cite{vernicos2018flag}. Other fields in which the Hilbert metric has been used are quantum information theory on convex cones defined by various operators\cite{reeb2011hilbertquantum}, machine learning in the form of clustering and graph embeddings \cite{nielsen2019clustering,nielsen2022nonlinear}, optimal mass transport \cite{chen2016entropic}, and a variety of situations in real analysis \cite{lemmens2013birkhoff}. Due to its many uses, various algorithms from classical computational geometry have been modified for use in the Hilbert metric, including Voronoi diagrams \cite{gezalyan2023voronoi, bumpus2023software} and Delaunay triangulations \cite{gezalyan2024delaunay}. We expand on these works by contributing an algorithm for minimum enclosing balls in the Hilbert polygonal geometry.

The Thompson metric was defined by A. C. Thompson in 1963 as an alternative to the Hilbert metric for its applications in analysis\cite{thompson1963certain}. Like the Hilbert, it provides a metric space over convex bodies and has a similar geometry. Its primary uses are in analysis, in particular as a metric on cones \cite{cobzacs2014normal,lemmens2015unique,lemmens2017midpoints,lim2013geometry}. We contribute a proof that minimum radius balls in this metric are LP-type. Additionally, we contribute the fact that Thompson balls are, like Macbeath Regions, also equivalent to Hilbert balls up to a scaling factor, and therefore induce the same topology on convex bodies.

The Funk weak metric was defined by Paul Funk in 1929 \cite{funk1929geometrien}. The Funk metric is a weak metric space that can be used to define both the Hilbert and the Thompson metrics, and is often studied in their context \cite{papadopoulos2014funk,troyanov2013funk} and in the context of flags of polytopes \cite{faifman2023volume}. Because the Funk metric is non-symmetric, it induces a reverse metric called the reverse Funk metric. We contribute a proof that the minimum radius ball problem in these weak metrics are LP-type.

\section{Preliminaries} \label{sec:prelims}
\subsection{Metric Spaces}

A metric space is the generalization of a distance $d:\mathbf{X}\times \mathbf{X} \rightarrow \RE^{\geq 0}$ on a set $\mathbf{X}$. It is a fundamental concept in geometry, real/complex analysis, and topology. We define it here:

\begin{definition}[Metric Space]
    The pair $(\mathbf{X},d)$ is a \textbf{metric space} if, for all $ a,b,c \in \mathbf{X}$:
    \begin{enumerate}
        \item $d(a,b) = 0$ iff $a=b$
        \item $d(a,b) = d(b,a)$
        \item $d(a,c)\leq d(a,b)+d(b,c)$
    \end{enumerate}
\end{definition}

When all of the above properties are satisfied except for symmetry, the space is called a \emph{weak metric}. Given a metric space, $(\mathbf{X},d)$, the closed ball of radius $r$ around a point $p$ generalizes the Euclidean idea of a circle to an arbitrary metric space. We define it here: 

\begin{definition} [Closed Ball]\label{def:ClosedBall}
A \textbf{closed ball} around a point $p$ in a metric space $(\mathbf{X},d)$ of radius $r$ is defined: 
\[B(p,r)=\{q\in \mathbf{X} \mid d(p,q)\leq r\}\] 
See Figure \ref{fig:FourBalls} for the closed balls in the four mentioned metrics.
\end{definition}
\label{def:MinimumRadiusBall}

A large set of metric spaces have a useful property known as the \textbf{Heine-Borel property}. A metric space is said to have the \textbf{Heine-Borel property} property if \emph{every closed and bounded set in the metric space is compact} \cite{williamson1987constructing}.

\subsection{The Hilbert and Thompson Metrics, and the Funk Weak Metric}

The Hilbert and Thompson metrics, and the Funk weak metrics, are defined on the points in the interior of a convex body $\Omega$ in $\RE^d$, where $\Omega$ is a closed and bounded convex set. Unless otherwise stated, we assume that $\Omega$ is a convex polygon with $m$ edges. We let $\bd \Omega$ refer to the boundary of $\Omega$. Given any pair of points $p$ and $q$ within the interior of $\Omega$, we define $\overline{p q}$ to be the cord of $\Omega$ through those two points, and $\chi(p,q)$ to be the directed ray from $p$ through $q$. Unless otherwise stated, when we take $p,q \in \Omega$, we mean the interior of $\Omega$. 

\begin{figure}[htbp]
\centerline{\includegraphics[scale=0.7]{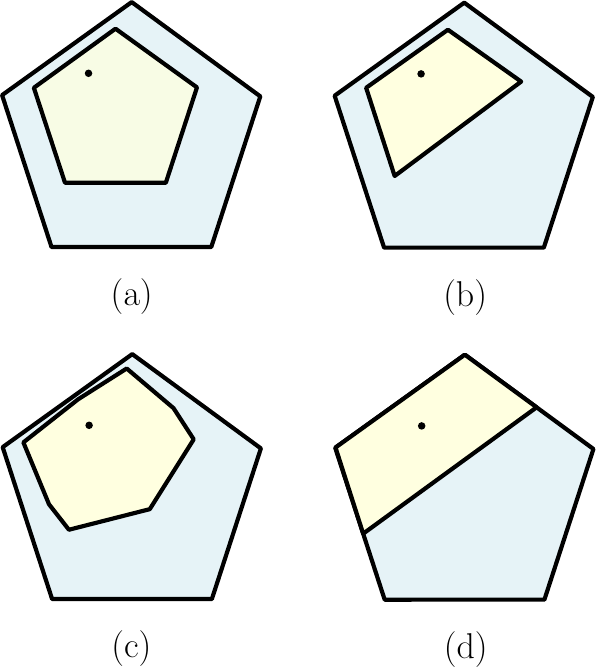}}
  \caption{(a) Funk ball, (b) Thompson ball (c) Hilbert ball, and (d) reverse Funk ball.}
  \label{fig:FourBalls}
\end{figure}

\begin{definition}[Funk weak metric]
Given two points $p,q$ in a convex polygon $\Omega$ in $\RE^n$ such that $\chi(p,q)$ intersects $\bd \Omega$ at a point $q'$ on $\bd\Omega$, we define the \textbf{Funk weak metric} to be:
\[
    F_\Omega(p,q)
         ~ = ~ \ln \frac{\|p - q'\|}{\|q - q'\|},
\]
where $F_\Omega(p,q)=0$.

\end{definition}

The above definition is also sometimes called the \emph{forward Funk metric}. Note however that the Funk weak metric is an asymmetric metric. Therefore, its reverse, the \emph{reverse Funk metric}, is defined to be $rF_\Omega(p,q)=F_\Omega(q,p)$. The \emph{Hilbert metric} can be defined as the average of the forward and reverse Funk metrics.

\begin{definition}[Hilbert metric] 
Given two distinct points $p,q$ in a convex polygon $\Omega$ in $\RE^n$, let $p'$ and $q'$ denote the endpoints of the chord $\overline{pq}$ on $\bd (\Omega)$, so that the points lie in order $\ang{p', p, q, q'}$, the \textbf{Hilbert distance} between $p$ and $q$, $H_\Omega(p,q)$, is:\[
    H_{\Omega}(p,q)
        ~ = ~ \frac{1}{2} \ln \left( \frac{\|q - p'\|}{\|p - p'\|} \frac{\|p - q'\|}{\|q - q'\|} \right),
\]
where $H_{\Omega}(p,p) = 0$.

\end{definition}

Since the product in the definition of the Hilbert metric is the cross ratio, which is preserved under projective transformations, it follows that the Hilbert metric is invariant under projective transformations\cite{nielsen2017balls}. It is also worth noting that straight lines are geodesics in the Hilbert metric, though not all geodesics are straight lines, and that the Hilbert metric satisfies all the properties of a metric, including symmetry and the triangle inequality. Nielsen and Shao showed how to algorithmically compute Hilbert balls with the aid of spokes. 

\begin{definition}[Spoke]
    Given $p \in \Omega$ a \textbf{spoke} through $p$ from a vertex $v$ of $\Omega$ is $\overline{pv}$.
\end{definition}
 
Succinctly, a Hilbert Ball of radius $r$ around a point $p$ can be constructed by computing all the points that are a distance of $r$ away from $p$ along the spokes of $p$ from the vertices of $\Omega$ and forming their convex hull (see Figure \ref{fig:FourBalls}(c)).

\begin{lemma}[Nielsen and Shao~\cite{nielsen2017balls}]
    Hilbert balls have at most $O(m)$ sides. 
\end{lemma}

The Thompson metric is similar to the Hilbert in its construction. It is the maximum of both the Funk and reverse Funk metrics. 

\begin{definition}[Thompson metric] 
Given a bounded closed convex body $\Omega$ in $\RE^n$ and two points $p,q \in \interior (\Omega)$, let $p'$ and $q'$ denote the endpoints of the chord $\overline{pq}$ on $\bd (\Omega)$, so that the points lie in order $\ang{p', p, q, q'}$, the \textbf{Thompson distance} between $p$ and $q$ $T_\Omega(p,q)$ is:\[
    T_{\Omega}(p,q)
        ~ = ~ max \left( \ln\frac{\|q - p'\|}{\|p - p'\|},\ln \frac{\|p - q'\|}{\|q - q'\|} \right),
\]
where $T_{\Omega}(p,p) = 0$.

\end{definition}

As it is the maximum of the two Funk weak metrics in a convex body, certain facts about this metric are clear. 

\begin{lemma}
    Balls in the Thompson metric are convex polygons with $O(m)$ sides.
\end{lemma}

\begin{proof}
    This follows directly from the fact that balls in the Funk metric are homotheties (the reverse is flipped) of $\Omega$ with $O(m)$ sides \cite{papadopoulos2014funk} (see Figure \ref{fig:FourBalls} (a) and (d)) and Thompson balls are their intersections (see Figure \ref{fig:FourBalls}(b)).
\end{proof}

\begin{theorem}
    The topology induced by the Hilbert or Thompson metrics, as well as the Funk weak metrics in a bounded convex domain $\Omega \subset \RE^n$ coincides with the Euclidean topology in that domain.
\end{theorem}

\begin{proof}
This is a well known result in the field. For reference, this follows from Theorem $2.1$ in \cite{beardon1999klein} where the authors show that the Hilbert and Euclidean metrics induce the same topology on bounded convex domains in $\RE^n$. In addition, refer to Proposition $6.1$ in "From Funk to Hilbert Geometry" for Funk weak metrics.  \end{proof}

From this we immediately have the following corollary. 
\begin{corollary}
    \label{lem:HilbertHB}
    Any Hilbert or Thompson metric space, as well as any Funk weak metric space, has the Heine-Borel property. 
\end{corollary} 

\section{Minimum Enclosing Radius
Balls and the Heine-Borel Property}

We begin by defining several concepts such as the minimum radius ball, and a set of criterion for being an LP-type problem\cite{sharir1992combinatorial}. 

\begin{figure}[htbp]
\centerline{\includegraphics[scale=0.7]{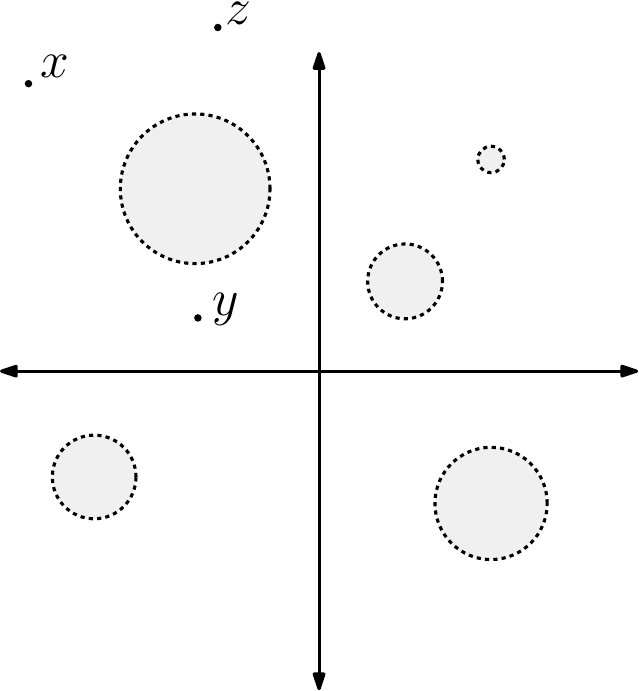}}
  \caption{This is $\RE^2$ with several closed balls removed. There is no Euclidean minimum radius ball through $x,y,z$.}
  \label{fig:RemovedDisks}
\end{figure}

\begin{definition}[Minimum Radius enclosing Ball]\label{def:minRadius}
    A closed ball $B$ is a \textbf{minimum radius enclosing ball} of $S \subset \mathbf{X}$ when $B$ has radius $\inf\{r\mid S\subset B(p,r)\}$. Note that in a general metric space, the  \textbf{minimum radius enclosing ball} need not exist, nor be unique.
\end{definition}

 \begin{definition}[Minimum Ball Property]
A metric space $(\mathbf{X}, d)$ satisfies the \textbf{minimum ball property} if for every finite $H\subset\mathbf{X}$ there exists a minimum radius enclosing ball of $H$.
 \end{definition}

 For an example of a metric space without the minimum ball property see Figure \ref{fig:RemovedDisks}. In this situation several closed disks were removed from $\RE^2$ with the Euclidean distance.
 
\begin{proposition}
\label{prop:HeineBorelMinimumBallProperty}
Let $(X, d)$ be a metric space. If $(X,d)$ satisfies the Heine-Borel property, then $(X,d)$ satisfies the minimum ball property. 
\end{proposition}
\begin{proof}
    For a given finite $H \subset \mathbf{X}$ define: 
    \[S := \{  s \in \mathbb{R}^{\ge 0} \mid \text{ there exists a closed ball of radius } s \text{ enclosing } H \}
    ~~\text{and}~~
    r := \inf S
    \]
    Note that $S$ is bounded below by 0 and is nonempty because $H$ is bounded as a consequence of being finite. So the infimum is defined and finite. Thus, to show there exists a ball of minimal radius, we want to show that there exists a ball of radius $r$ enclosing $H$. 
Define the set of centers $C_n := \{ p \in \mathbf{X} : H \subset B(p, r + 1/n)\}$. For all $n \in \mathbb{N}$ we have $C_n \supset C_{n+1}$ because for all $p \in C_{n+1}$ we have $H \subset B(p, r + 1/(n+1)) \subset B(p, r+ 1/n)$ and thus $p \in C_n$ as well. 

Each $C_n$ is bounded since for any $h \in H$ and any $p \in C_n$, $d(p,h) \le r+1/n$. 

We will now show that each $C_n$ is closed. Let $h\in H$ and $\{p_i\}_{i \in \mathbb{N}} \subset C_n$  such that $p_i \rightarrow p$, $p \in \mathbf{X}$. Then for any $\varepsilon >0$, there exists $i$ such that $d(p,p_i)<\varepsilon$. So by the triangle inequality, $d(p, h) \le d(p,p_i)+d(p_i,h)  \le (r+1/n)+\varepsilon$ for any $\varepsilon$, so $d(p,h) \le r+1/n$. So $h \in B(p, r+ 1/n)$ for any $h \in H$. Thus $p \in C_n$. So $C_n$ contains all of its limit points and is thus closed. 

Define $D: = \bigcap_n C_n$. Since we are assuming the Heine-Borel property and showed each $C_n$ is closed and bounded, each $C_n$ is compact. Because $\{C_n\}$ is a decreasing nested sequence of compact sets of $\mathbf{X}$ with the metric topology, by Cantor's intersection theorem $D$ is nonempty \cite{rudin1964principles}. 

Fix $p \in D$. Then for any $h \in H$, $d(p,h) \le r+1/n$ for all $n\in \mathbb{N}$. Thus $d(p,h) \le r$. We conclude $H \subset B(p,r)$. 
\end{proof}
Note that the converse does not hold. As a counterexample, any infinite set equipped with the discrete metric does not satisfy the Heine-Borel property, yet it satisfies the minimum ball property.
Additionally, note that in the above proof of Theorem \ref{prop:HeineBorelMinimumBallProperty} we did not use the symmetric property of the metric space, only the triangle inequality. As such, given the same definition of a ball Definition \ref{def:ClosedBall} and Definition \ref{def:minRadius} e.g. fixing the direction of the ball, we contribute the following corollary. 

\begin{corollary}\label{cor:WeakHeine-Borel}
    Let $(X, d)$ be a weak metric space. If $(X,d)$ satisfies the Heine-Borel property, then $(X,d)$ satisfies the minimum ball property. 
\end{corollary}

We now introduce the characterization of an LP-type problem by Sharir and Welzl \cite{sharir1992combinatorial}.
   
\begin{definition} [LP-type]\label{def:LPtype}
    A pair $(H,f)$ is called an \textbf{LP-type problem} if $H$ is a finite set and $f$ is a function from subsets of $H$ to a totally ordered set such that
     $f$ satisfies the following two properties:
    \[Monotonicity: F\subset G \subset H \implies f(F)\leq f(G) \leq f(H).
    \]
    \[    Locality: F \subset G \subset H, x \in H, f(F) = f(G) = f(F \cup \{x\}) \implies f(F) = f(G \cup \{x\}).
    \]

\end{definition}

    Let $(\mathbf{X}, d)$ be a metric space satisfying the minimum ball property. Let $\mathbf{X}$ have some well-ordering $\preceq$. Take a finite set $H \subset \mathbf{X}$. We define:  
\[
f:2^{H} \rightarrow \RE \times \mathbf{X}, f(G)=\inf\{(r,p) \mid  G \subset B(p,r), \text{ordered lexicographically}\}
\]

Here, $f(G) = (\hat{r}, \hat{p})$ gives us a unique minimum radius enclosing ball of $G$, by taking the minimum center with respect to the well-ordering. We call this ball $B_G = B(\hat{p}, \hat{r})$. Note that if $f(G) = (r, p)$, there may exist other points $p'$ such that $G \subset B(r,p')$ as well; however, $f$ defines only one minimum radius enclosing ball. Note also that if $\mathbf{X} = \RE^n$ then by Proposition \ref{prop:WelltoLexOrder} we can take the lexicographic usual total order and $f$ will still be well-defined.

\begin{theorem}\label{thrm:MetricSpaceLP}
    The problem of finding the minimum radius ball of a set of points in a metric space that satisfies the minimum ball property is LP-type.
\end{theorem}

\begin{proof}
To see that $f$ satisfies \emph{monotonicity}, consider a finite set $H \subset \mathbf{X}$. Let $F\subset G\subset H$ and suppose $f(H) = (\hat{r}, \hat{p})$. Note that $G \subset H \subset B_H$ and thus $(\hat{r}, \hat{p}) \in \{(r, p) \mid G \subset B(p,r)\}$. It follows from the definition of $f$ that $f(G) \leq f(H)$. By the same logic, $f(F)\leq f(G)$.

To see that $f$ satisfies \emph{locality}, consider a finite set $H \subset \mathbf{X}$. Take $F,G \subset H$ and $x \in H$ such that $F\subset G\subset H$ and $f(F) = f(G) = f(F \cup \{x\})$. It follows that $B_F = B_G = B_{F \cup \{x\}}$. Hence, $x\in B_G$ and $ F \subset G \cup \{x\} \subset B_G = B_F$. Thus, by monotonicity, $f(F) \le f(G \cup \{x\}) \le f(G) = f(F)$, and we conclude that $f(F) = f(G \cup \{x\})$. 
\end{proof}

Note that the only properties of the metric space used for Theorem \ref{thrm:MetricSpaceLP} is the minimum ball property, which we've shown can also hold for weak metric spaces in Corollary \ref{cor:WeakHeine-Borel}, and that if $G \subset H$ then $G \subset B_H$, which still holds for weak metrics so long as the direction of the ball is fixed. As such, we contribute the following corollary.

\begin{corollary} \label{cor:WeakLp}
    Finding the minimum radius ball of a set of points in a weak metric space that satisfies the minimum ball property is LP-type.
\end{corollary}

Note that a well-ordering is not always necessary and can be replaced with a total ordering, such as a lexicographic ordering. 

\begin{proposition}
\label{prop:WelltoLexOrder}
    If $X = \RE^k$  for $n < \infty$ then the well-order $\preceq$ on $\mathbf{X}$ can be replaced with the lexicographical usual total order $\le$ from $\RE^k$. 
    \end{proposition}

\begin{proof}

Recall, the following notation and results from Theorem \ref{prop:HeineBorelMinimumBallProperty}. 

$$C_n := \{ p \in \mathbf{X} : H \subset B(p, r + 1/n)\} \qquad\text{and}\qquad D: = \bigcap_n C_n$$ 

We saw that each $C_n$ is closed and bounded with respect to the metric topology. Thus it is compact in the standard (metric) topology in $\mathbb{R}^k$, and thus $D$, the set of centers of minimum radius enclosing balls is also nonempty and compact in this topology. 

Let $p_i: \RE^k \rightarrow \RE^d$ be the projection onto the $i^{th}$  component. Since $D$ is compact in the standard topology, its projection onto each component is compact and thus compact in the standard topology on $\RE$. We want to show that there is a minimal point with respect to the lexicographic order on $\RE^k$. 

 The body $p_1(D)$ is compact, and thus it has a minimal element, $x_1$.

Then $D_1: =p_1^{-1}(x_1) \bigcap D $ is non-empty and compact (in standard topology) since $D$ is compact and $p_1^{-1}(x_1)$ is closed.

We similarly recursively define $x_j $ to be the minimal element of $p_j(D_{j-1})$ and find $D_j:=p_j^{-1}(x_j) \bigcap D_{j-1}$ to be compact for $j \in \{2,..,k\}$. It follows that $(x_1, x_2, ...x_k)$ is the minimal element of $D$ with respect to the lexicographic order.  \end{proof}

\begin{corollary}
\label{cor:problem_is_LP_allmetrics}
    The minimum radius ball problem in the Hilbert and Thompson metrics, as well as the Funk weak metrics, is LP-type.
\end{corollary}

\begin{proof}
Note that Corollary \ref{lem:HilbertHB}, Proposition \ref{prop:HeineBorelMinimumBallProperty}, and Theorem \ref{thrm:MetricSpaceLP} together imply this result for the Hilbert and Thompson metrics. Note secondly that Corollary 10, Corollary \ref{cor:WeakHeine-Borel}, and Corollary \ref{cor:WeakLp} together imply the result for the Funk weak metrics.
\end{proof}

\section{Hilbert Radius Minimum Enclosing Ball}

We have shown in Corollary \ref{cor:problem_is_LP_allmetrics} that the minimum radius enclosing ball problem is LP-type in a Hilbert metric, so we will now focus on how to implement the LP algorithm. By 
Proposition \ref{prop:WelltoLexOrder} the well-order $\preceq$ on $\mathbf{X}$ can be replaced with the lexicographical usual total order $\le$ from $\RE^k$. This is useful as we can comprehend and compute using the lexicographic total order.

Running an LP algorithm on an LP-type problem requires solving two primitive operations \cite{matouvsek1992subexponential}. These are the \textbf{violation test} (given a basis $B\subset S$ and element $x\in S$, whether $f(B)=f(B\cup \{x\})$) and \textbf{basis computations} (how to find a basis of $B\cup \{x\}$). Instead of a well-order, we will let $\delta$ be the lexicographic order $\RE^2$ for the Hilbert convex polygons. We need the two following supporting lemmas:

\begin{lemma}
    The combinatorial dimension for minimum Hilbert radius balls is 3.
\end{lemma}

\begin{proof}
    This follows immediately from the fact that, in general position at most three points defines a Hilbert ball \cite{gezalyan2024delaunay}(see Lemma 14). This is because Hilbert balls intersect along line segment edges in general position. Three points cannot intersect along the same edge at the same distance without two of them lying along the same spoke.
\end{proof}

\begin{lemma}\label{lem:TwoPointCenter}
    Given an $m$ sided convex polygon $\Omega$ in $\RE^2$, the center of the minimum Hilbert radius ball around two points, $p,q \in \interior \Omega$, can be computed in time $O(\log m)$. 
\end{lemma}

\begin{proof}
    The balls of minimum radius for a set of two points are always centered on a section of the bisector in the sector directly between the two points. This is because Hilbert balls are polygons and the balls around two points meet at a segment across this sector\cite{nielsen2017balls} (see Lemma $10$ page $4$). This piece of the bisector is a line \cite{bumpus2023software}. We binary search the boundary of $\Omega$ to determine which edges define this sector, and then use the bisector equation there \cite{gezalyan2024delaunay,bumpus2023software} (see Section 3 in \cite{gezalyan2024delaunay}). We choose the lexicographically smallest point along the bisector in this sector to serve as the center of the minimum Hilbert radius ball around points $p, q$.
\end{proof}

\begin{lemma}\label{lem:violationTest}
    The \textbf{violation test} in the Hilbert metric for radius balls can be computed in time $O(\log^3 m)$.
\end{lemma}

\begin{proof}
    Given a basis $B$, we are interested in if $f(B)=f(B\cup \{x\})$. We assume that $x$ is not already in $B$. We have two nontrivial cases based on the size of the basis.

    Case 1: $B$ has two elements $y,z$. To determine if $f(B) = f(B \cup \{x\})$, it suffices to check whether $x$ is contained in the minimum radius ball of $y,z$.  If $x$ is contained, then $f(B) = f(B \cup \{x\})$. If $x$ is not contained, either the center must move or the radius must increase, so $f(B) < f(B \cup \{x\})$. Since we can compute the center of a $2$ point Hilbert ball in $O(\log m)$ time by Lemma \ref{lem:TwoPointCenter}, we can check this case in $O(\log m)$ time.

    Case 2: $B$ has three elements. In this case, we calculate the center of the Hilbert ball around the three points in time $O(\log^3 m)$ using the algorithm from "Delaunay Triangulations in the Hilbert Metric" \cite{gezalyan2024delaunay}. If a center exists, we compute the distance between $x$ and the center. If a center does not exist, one of the three elements is contained in a ball defined by the other two, so we check all pairs of elements using Lemma \ref{lem:TwoPointCenter} in $O(\log m)$ time to find the minimum enclosing ball. We take the resulting ball and check the center's distance to $x$. Either $x$ is in the ball or it is not, in which case $f(B)$ must increase because either the radius or center must move.
\end{proof}

\begin{lemma}
    The \textbf{basis computation} in the Hilbert metric for radius balls can be done in time $O(\log^3 m)$.
\end{lemma}

\begin{proof}
    Suppose we have a previous basis $B$ and a new element $x$ that is not contained inside the ball formed by the basis. We would like to compute the basis of $B \cup \{x\}$. This gives us a few cases. We will consider the non-trivial cases where $B$ contains two or three elements. In either situation, we compute the minimum enclosing balls of all pairs and triples of points in $B\cup \{x\}$, and check for containment of the remaining points in those balls. This can be done in $O(\log^3 m)$ as described in Case 3 of Lemma \ref{lem:violationTest}.   
\end{proof}

\begin{theorem}
    Hilbert minimum radius balls, of $n$ points, can be computed in time $O(n\log^3 m)$.
\end{theorem}

\begin{proof}
    This follows directly from the running time of our primitive operations and Theorem 7 and Corollary 8 from Sharir and Welzl's "A combinatorial bound for linear programming and related problems"\cite{sharir1992combinatorial}.
\end{proof}

\section{Conclusion}
In this paper, we contributed a criterion for showing that minimum radius balls are LP-type. This criterion being that finite sets in the metric space always have at least one minimum radius ball. We showed that if a metric space, or weak metric space, has the Heine-Borel property, it has this property. We used this to contribute a minimum radius ball algorithm for the Hilbert metric and proved that the minimum radius ball problem is LP-type for the Thompson metric and Funk weak metrics. 

\bibliography{shortcuts,hilbert}

\end{document}